\documentclass[oneside,reqno,english]{amsart}
\usepackage[T1]{fontenc}
\usepackage[latin9]{inputenc}
\usepackage{units}
\usepackage{amsthm}
\usepackage{amssymb}
\usepackage{setspace}

\makeatletter
\numberwithin{equation}{section}
\numberwithin{figure}{section}
\theoremstyle{definition}
\newtheorem*{defn*}{\protect\definitionname}
\theoremstyle{remark}
\newtheorem*{rem*}{\protect\remarkname}
\theoremstyle{plain}
\newtheorem{thm}{\protect\theoremname}
\theoremstyle{definition}
\newtheorem{defn}[thm]{\protect\definitionname}
\theoremstyle{remark}
\newtheorem{rem}[thm]{\protect\remarkname}
\theoremstyle{plain}
\newtheorem{lem}[thm]{\protect\lemmaname}

\makeatother

\usepackage{babel}
\providecommand{\definitionname}{Definition}
\providecommand{\lemmaname}{Lemma}
\providecommand{\remarkname}{Remark}
\providecommand{\theoremname}{Theorem}

\begin{document}
\title{{\small{}Celestial Mechanics Solutions where the Future is a Perfect
Reflection of the Past}}
\author{{\small{}Ali Abdulhussein$^{1}$ and Harry Gingold$^{2}$}}
\address{{\small{}Department of Mathematics , WVU, Morgantown WV 26506}}
\email{aaabdulhussein@mix.wvu.edu{\small{}$^{1}$ }, gingold@math.wvu.edu{\small{}$^{2}$}}
\begin{abstract}
Newton's equations of celestial mechanics are shown to possess a continuum
of solutions in which the future trajectories of the $N$ bodies are
a perfect reflection of their past. These solutions evolve from zero
initial velocities of the $N$ bodies. Consequently, the future gravitational
forces acting on the $N$ bodies are also a perfect reflection of
their past. The proof is carried out via Taylor series expansions.
A perturbed system of equations of the $N$ body problem is also considered.
All real valued solutions of this perturbed system have no singularities
on the real line. The perturbed system is shown to have a continuum
of solutions that possess symmetry where the future velocities of
the $N$ bodies are a perfect reflection of their past. The positions
and accelerations of the $N$ bodies are then odd functions of the
time. All $N$ bodies then evolve from one location in space.
\end{abstract}

\maketitle

\section{{\small{}Introduction.}}

{\small{}Solutions that possess symmetry are of great interest in
celestial mechanics. Periodic motion of celestial bodies may be considered
a form of symmetry that is central to our understanding of the solar
system. }Given a periodic function $y(t)$ with a period $P>0\thinspace,n\in\mathbb{N}$,
then the following identity in $t$ is a manifestation of symmetry
\begin{equation}
y(t-nP)\equiv y(t+nP).\label{eq:PERIODICTYIDENTITY}
\end{equation}
{\small{} demonstration of the predictive powers of Newtons' equations
of celestial mechanics is the existence of solutions that are conic
sections that lie in one plane. This, for the Kepler two body problem.
}E.g. \cite{POLARDCELESTIAL-1} . Even functions and odd functions
respectively, are also a manifestation of symmetry. Namely, 
\begin{equation}
y(t)=y(-t),\,-y(t)=y(-t).\label{eq:EVENODDsymmetry}
\end{equation}

{\small{}In this atrticle we look for celestial mechanics solutions
that possess symmetries other than periodicity. We answer in the affirmative
the question whether or not the $N$ body problem }possesses solutions
in which the future is a perfect reflection of the past. This is so
iff the initial velocities of all $N$ bodies is zero. Consequently,
the future gravitational forces acting on the $N$ bodies are also
a perfect reflection of their past. The derivation of these conclusions
point to another interesting aspect of the predictive powers of Newtons'
equations of celectial mechanics. We don't know for sure what is the
past the present and the future of the universe. Are there astronomically
remote subsystems of point massess that approximately possess symmetries
that Newton's equations predict? This work is also a contribution
to our interest in past present and future of the universe. Compare
with \cite{Brumberg-1} . A perturbed system of equations of the $N$
body problem is also considered. All real valued solutions of this
perturbed system have no singularities on the real line. The perturbed
system is shown to have a continuum of solutions that possess symmetry
where the future velocities of the $N$ bodies are a perfect reflection
of their past. The positions and accelerations of the $N$ bodies
are then odd functions of the time. This is so iff the $N$ bodies
evolve from one location in space. We provide a proof via Taylor series
expansions. This method is tangible to scientists. The method of proof
guarantees the existence of real analytic even solutions to the $N$
body problem. It also guarantees the existence of even and odd analytic
solutions to an approximate model of the equations. For analytic solutions
of differential systems compare e.g. with \cite{HILLEODESCOMPLEX-1,KRANTZRealAnalytic,P.F.Hsieh=000026Y. Sibuya-2}. 

We prove two theorems that apply to a significant list of second order
nonlinear vectorial autonomous differential equations. The list includes
centarl force problems, like the Manev problem, and the Pendudlum
equation. Compare e.g. with \cite{Boyce E.W. and DiPrima R.C-1,BRAUERNOHEL-1,DIACUISENTROPIC-1,GuckenheimerHolmes-1,JORDAN=000026SMITH-1,L. Perko-1,MANEVDIACUII-1,MANEVI-1}.
They could be also useful to numerical approximations of solutions
and to phase space analysis . Given a scalar ordinary differential
equation $y''=f(y)$, let $(y,y')$ be the phase plane. Then, the
even solutions orbits intercept the $y$ axis and the odd solutions
orbits, if any, intercept the $y'$ axis. It is not easy to construct
and visualize global phase diagram for a $2n$ dimensions space in.
However, analiticity as proven in here could provide guidance to some
local phase space analysis for every pair of coordinate $(y_{j},y_{j}')$.
The order of presentations in this article is as follows. In section
2 we discuss preliminary notations and conventions and we motivate
an approximation model to the $N$ body problem equations. In section
3 we prove one of the main theorems in this article about the existance
of even solutions to the celestial mechanics solutions. In section
4 we provide a discussion about the meaning of $-f(y)=f(-y)$ and
we provide a lemma on the symmetry of partial derivatives of $f(y)$
. These, are necessary for the proof of our second main theorem on
the existence of odd solutions to an approximation model to the $N$
body problem equations. In section 5 we provide a lemma that presents
two successive even order derivatives of the components of $y$ .In
section 6 we provided a proof that the approximation model possesses
a continuum of odd solutions.

\section{Preliminary notations and conventions and an approximation model}

Some of the notation below is motivated by the necessity to formulate
initial value problems for the celestial mechanics equations that
are complex valued and exists in disks of the complex plane $t$.
However, as will be seen in the sequel, the initial point and the
initial positions and velocities are required to be real valued. The
following notation is used, $m,r,s,j,k,l,\lambda,s_{l},z_{l},k_{l},N\in\mathbb{N}_{0}$,
$\mathbb{N}_{0}$ is the set of nonnegative integers :$\thinspace m_{1},m_{2},\ldots,m_{N}$
are the masses of the $N$ bodies; $\thinspace t\in\mathbb{C}$ is
the time variable; $y_{j}\in\mathbb{\mathbb{C}}^{3}$ where $1\leq j\leq N$,
are column position vectors of the $N$ bodies, respectively; $T$
stands for transposition of a vector or a matrix; $y^{T}=[y_{1,}y_{2,},\ldots,y_{N}]$
and $f(y)^{T}:=[f_{1}(y),f_{2}(y),\ldots,f_{N}(y)]$ are respectively
rows of blocks of column vectors; $\overline{y}_{j}$ is the complex
conjugate of $y_{j}$ ; $y_{k}(t_{0}),\,y_{j}(t_{0})\in\mathbb{R}^{3},\;y_{k}(t_{0})\neq y_{j}(t_{0}),\;k\neq j,\;k,j=1,2,\ldots,N.$
Let 
\begin{equation}
D_{\epsilon}(t_{0}):=\{t\,\bigr|\left|t-t_{0}\right|\leq\epsilon,\;\epsilon>0,\;t_{0}\in\mathbb{R},t\in\mathbb{\mathbb{C}}\}.\label{eq:SmallDisk}
\end{equation}
We also adopt the following definition.
\begin{defn*}
Given $A$ an $m$ by $n$ matrix and $y$ an $n$ by 1 column vector
with elements in $\mathbb{C}$. A norm denoted by $\left|.\right|$
is called algebraic if it satifies the following inequality
\begin{equation}
\bigl|Ay\bigr|\leq\bigl|A\bigr|\bigl|y\bigr|.\label{eq:ALGEBRAICNORM}
\end{equation}
The notation $\left\Vert y\right\Vert $ that normally stands for
$\left\Vert y\right\Vert =[y_{j}^{T}\overline{y}_{j}]^{\frac{1}{2}}$that
is the Euclidean norm for $y$ complex , is replaced in here by the
unconventional use

\begin{equation}
\left\Vert y_{j}\right\Vert :=[y_{j}^{T}y_{j}]^{\frac{1}{2}},j=1,2,\ldots,N,\:y_{j}\in\mathbb{\mathbb{C}}^{3}.\label{eq:EuclideanNorm}
\end{equation}
\end{defn*}
\begin{rem*}
The unconventional use of notation (\ref{eq:EuclideanNorm}) requires
justification. We desire to prove the existence of analytic solutions
to the celestial mechanics equations with (\ref{eq:f(y)NOTATION})
, (\ref{eq:NBODY2NDORDERVEC}), and with (\ref{eq:fk(y)hat}), (\ref{eq:IVPAPPROXIMATENBODY}).
The convention (\ref{eq:EuclideanNorm}) together with the real values
of the initial point and the initial values guarantees that $f_{k}(y)$
and $\widehat{f_{k}}(y)$ are analytic vector functions of $y$ in
some disk in $\mathbb{\mathbb{C}}^{3N}$. More specifically, the denominators
in (\ref{eq:f(y)NOTATION}) and in (\ref{eq:fk(y)hat}) stay analytic
in some disk in $\mathbb{\mathbb{C}}^{3N}$ and are bounded away from
zero.
\end{rem*}
Put,
\begin{equation}
f_{k}(y):=\sum_{j=1,j\neq k}^{N}\frac{Gm_{j}(y_{j}-y_{k})}{\|y_{k}-y_{j}\|^{3}},\;1\leq k\leq N.\label{eq:f(y)NOTATION}
\end{equation}
Denote $\nicefrac{dy_{k}}{dt}=y_{k}'$ , $\nicefrac{d^{2}y_{k}}{dt^{2}}=y_{k}'',\:\nicefrac{d^{l}y_{k}}{dt^{l}}=y_{k}^{(l)},etc$,
$t_{0}\in\mathbb{R},y_{0},\eta\in\mathbb{\mathbb{\mathbb{R}}}^{3N}$.
Then, the initial value problem for the $N$ bodies is

\begin{equation}
y''=f(y),\;y(t_{0})=y_{0},y'(t_{0})=\eta,\;y_{k}(t_{0})\neq y_{j}(t_{0}),\;k\neq j,\;k,j=1,2,\ldots,N.\label{eq:NBODY2NDORDERVEC}
\end{equation}
{\small{}Observe that Newtons equations of celestial mechanics satisfy
\begin{equation}
-f_{k}(-y)=-\sum_{j\ne k}\frac{Gm_{j}(-y_{j}+y_{k})}{\|y_{j}-y_{k}\|^{3}}=\sum_{j\ne k}\frac{Gm_{j}(y_{j}-y_{k})}{\|y_{j}-y_{k}\|^{3}}\Longrightarrow-f(y)=f(-y).\label{eq:NEWTON-f(y)=00003Df(-y).}
\end{equation}
However, we cannot solve for odd solutions the initial value problem
(\ref{eq:NBODY2NDORDERVEC}) with a condition $y(t_{0})=\overrightarrow{0}$
. This, because $y(t_{0})=\overrightarrow{0}$ means that all of the
celestial point masses are in state of mutual collision. Then, the
initial value problem (\ref{eq:NBODY2NDORDERVEC}) contains undetermined
and unbounded terms which render the equations invalid. Therefore,
we consider a modified celestial mechanics system of equations with
$\epsilon(j,k)$ small. Namely,
\begin{equation}
\widehat{f}_{k}(y):=\sum_{j=1,j\neq k}^{N}\frac{Gm_{j}(y_{j}-y_{k})}{[\|y_{k}-y_{j}\|+\epsilon(j,k)]^{3}},\;1\leq j,k\leq N,\;\epsilon(j,k)=\epsilon(k,j)>0.\label{eq:fk(y)hat}
\end{equation}
}It is easily verified that $-\widehat{f_{k}}(y)=\widehat{f_{k}}(-y)$
as well. Consequently, with 
\begin{equation}
-\widehat{f}(y)^{T}:=[\widehat{f}_{1}(y),\widehat{f}_{2}(y),\ldots,\widehat{f}_{N}(y)]^{T}=\widehat{f}(-y)^{T},\label{eq:-fHat(y)=00003DfHat(-y).}
\end{equation}
we obtain an initial value problem for an approximated equation 
\begin{equation}
y''=\widehat{f}(y),\;y(t_{0})=y_{0},y'(t_{0})=\eta,\;y_{0},\eta\in\mathbb{\mathbb{\mathbb{R}}}^{3N}.\label{eq:IVPAPPROXIMATENBODY}
\end{equation}
Formally, we also have that the limits of all $\widehat{f}_{k}(y)$
with all $\epsilon(j,k)\rightarrow0^{+}$ are respectively $f_{k}(y)$
of Newton's equations (\ref{eq:f(y)NOTATION}). This initial value
problem (\ref{eq:IVPAPPROXIMATENBODY}) can be solved for any $y_{0},\eta\in\mathbb{\mathbb{\mathbb{R}}}^{3N}$.
Notice that 
\[
y_{0},\eta\in\mathbb{\mathbb{\mathbb{R}}}^{3N},\;t,t_{0}\in\mathbb{R}\Longrightarrow y(t)\in\mathbb{\mathbb{\mathbb{R}}}^{3N}.
\]
Also observe that for $y(t)\in\mathbb{\mathbb{\mathbb{R}}}^{3N}$,
a solution to the initial value problem (\ref{eq:IVPAPPROXIMATENBODY})
, we have 
\begin{equation}
\left\Vert y''(t)\right\Vert =\left\Vert \widehat{f}(y)\right\Vert \leq\left\Vert \sum_{j=1,j\neq k}^{N}\frac{Gm_{j}(y_{j}-y_{k})}{[\|y_{k}-y_{j}\|+\epsilon(j,k)]^{3}}\right\Vert \leqq NG[\max(m_{j})]\max(\frac{1}{[\epsilon(j,k)]^{2}}).\label{eq:BOUNDONy''DERAPPROX}
\end{equation}
 $Max(m_{j})$ and $\max(\nicefrac{1}{[\epsilon(j,k)]^{2}})$ are
taken over all $j=1,2,\cdots,N.$ Consequently, all solutions of (\ref{eq:IVPAPPROXIMATENBODY})
exist on $(-\infty,\infty)$ . Thus, the system of equations $y''=\widehat{f}(y)$
is (so called) complete. However, solutions of (\ref{eq:IVPAPPROXIMATENBODY})
as analytic functions of $t$ , could develop singularities in the
complex plane.
\begin{rem*}
The formulation of the two main theorems in the sequel assume $y\,,f(y)\in\mathbb{C}^{n},\;n\in\mathbb{N}$,
that are somewhat more general than the $f(y)$ and $\widehat{f}(y)$
discussed above for which $n=3N$ was restricted. 
\end{rem*}

\section{Formulation and proof of main theorem 1 by induction on odd derivatives}
\begin{thm}
Assume that: i) $t\in\mathbb{C},\:t_{0}\in\mathbb{R},y_{0},\eta\in\mathbb{\mathbb{\mathbb{R}}}^{3N},y\,,f(y)\in\mathbb{C}^{n},\;n\in\mathbb{N}$,
where $f(y)$ is an analytic function in the vector variable $y$
in a disk such that in $D$
\begin{equation}
D:=\{y\bigl|\left\Vert y-y_{0}\right\Vert \leq b\}\Longrightarrow\left\Vert f(y)\right\Vert \leq M.\label{eq:Boundedf(y)in a disk.}
\end{equation}
 Then, the initial value problem

\begin{equation}
y''=f(y),\;y(t_{0})=y_{0},y'(t_{0})=\overrightarrow{0},\;\overrightarrow{0}^{T}:=[0,0,\cdots,0],\label{eq:INITIALVALUEPROBGENERAL}
\end{equation}
possesses a unique analytic solution $y(t)$ for $\left|t-t_{0}\right|\leq\sqrt{\nicefrac{2b}{M}}$
that satisfies $y(t-t_{0})\equiv y(-(t-t_{0}))$. Namely, $y^{(m)}(t_{0})=\overrightarrow{0}$
for all odd numbers m.
\end{thm}

\begin{proof}
Without loss of generality assume that $t_{0}=0$ since our differential
system is autonomous. Below are two successive odd order derivatives.
Compare also with \cite{ABROGASTDERCOMPOSITESCALA-1,diBRUNO-1,ENCIASSHORTPROOFBRUNO}.

{\small{}
\begin{equation}
\frac{d^{3}y_{j}(t)}{dt^{3}}=\sum_{k_{1}=1}^{n}\frac{\partial f_{j}(y(t))}{\partial y_{k_{1}}}\frac{dy_{k_{1}}(t)}{dt}\thinspace;j=1,2,\ldots,n.\label{eq:3RDERyj}
\end{equation}
}{\small\par}

{\small{}
\[
\frac{d^{5}y_{j}(t)}{dt^{5}}=\sum_{k_{1}=1}^{n}\frac{\partial f_{j}(y(t))}{\partial y_{k_{1}}}\frac{d^{3}y_{k_{1}}(t)}{dt^{3}}+\sum_{k_{1}=1}^{n}\sum_{k_{2}=1}^{n}\frac{\partial^{2}f_{j}(y(t))}{\partial y_{k_{2}}\partial y_{k_{1}}}[2\frac{dy_{k_{2}}(t)}{dt}\frac{d^{2}y_{k_{1}}(t)}{dt^{2}}+\frac{d^{2}y_{k_{2}}(t)}{dt^{2}}\frac{dy_{k_{1}}(t)}{dt}]
\]
 
\begin{equation}
+\sum_{k_{1}=1}^{n}\sum_{k_{2}=1}^{n}\sum_{k_{3}=1}^{n}\frac{\partial^{3}f_{j}(y(t))}{\partial y_{k_{3}}\partial y_{k_{2}}\partial y_{k_{1}}}\frac{dy_{k_{3}}(t)}{dt}\frac{dy_{k_{2}}(t)}{dt}\frac{dy_{k_{1}}(t)}{dt}.\label{eq:5THDERyj}
\end{equation}
}It is easily verified that with $y^{T}=[y_{1,}y_{2,},\ldots,y_{j},\ldots,y_{n}]^{T}$
and with $y_{j}(t)$ scalars, we get from (\ref{eq:3RDERyj}) and
(\ref{eq:5THDERyj}) 
\begin{equation}
y'(0)=\overrightarrow{0}\Longrightarrow y^{(3)}(0)=y^{(5)}(0)=\overrightarrow{0}.\label{eq:DeR3rd5thzero}
\end{equation}
These particular cases (\ref{eq:3RDERyj}) and (\ref{eq:5THDERyj})
indicate what should be the general form of higher odd order derivatives
of $y_{j}(t)$ . Moreover, they demonstrate how the property of zero
initial velocities $y_{j}^{(1)}(0)=0$ is inherited by subsequent
derivatives of odd order. In what follows we may suppress the notation
$(t)$ in $y(t),y_{j}^{(\lambda)}(t)$ etc, when clarity is not compromised.
Assume that each component $y_{j}^{(2+m)}(t)$, $m$ odd, $j=1,2,\cdots,n$
, is a finite sum of products of the form:
\begin{equation}
T_{m}:=\frac{\partial^{l}f_{j}(y(t))}{\partial y_{k_{l}}\ldots\partial y_{k_{2}}\partial y_{k_{1}}}J_{O}J_{E},\label{eq:ODDmDERIVATIVE}
\end{equation}
i) $J_{O}$ is a finite product of an odd number $r\in\mathbb{N}_{0}$
of odd order derivatives of certain components of $y(t)$. Namely,
\begin{equation}
J_{O}=y_{s_{1}}^{(2e_{`1}+1)}(t)y_{s_{2}}^{(2e_{2}+1)}(t)\ldots y_{s_{r}}^{(2e_{r}+1)}(t),\qquad e_{1},e_{2},\ldots,e_{r}\in\mathbb{N}_{0},\:r\geq1.\label{eq:ODDDERIVATIVES=000026ODDNUMBER}
\end{equation}
ii) $J_{E}$ is a finite product of any number $w\in\mathbb{N}_{0}$
of even order derivatives of components of $y(t)$. Namely,
\begin{equation}
J_{E}=y_{z_{1}}^{(2c_{1})}(t)y_{z_{2}}^{(2c_{2})}(t)\ldots y_{z_{w}}^{(2c_{w})}(t),\qquad c_{1},c_{2},\ldots,c_{w}\in\mathbb{N}_{0}.\label{eq:ANYNUMBEREVENDERIVATIVES}
\end{equation}
If $w=0$, we put $J_{E}\equiv1$. Then, $y_{j}^{(4+m)}(t)$ is a
finite sum of certain products of the form

\begin{equation}
\widehat{T_{m}^{(2)}}:=\frac{\partial^{s}f_{j}(y(t))}{\partial y_{k_{s}}\ldots\partial y_{k_{2}}\partial y_{k_{1}}}\widehat{J}_{O}\widehat{J}_{E},\label{eq:ODDmDERIVATIVE+2}
\end{equation}
where $\widehat{J}_{O}$ is a finite product of an odd number $p\in\mathbb{N}$
of odd order derivatives of components of $y(t)$. Namely,
\begin{equation}
\widehat{J}_{O}=y_{s_{1}}^{(2g_{1}+1)}(t)y_{s_{2}}^{(2g_{2}+1)}(t)\ldots y_{s_{p}}^{(2g_{p}+1)}(t),\quad g_{1},g_{2},\ldots,g_{p}\in\mathbb{N}_{0},\:p\geq1,\label{eq:ODDDERIVATIVES=000026ODDNUMhat}
\end{equation}
and $\widehat{J}_{E}$ is a finite product of any number $q\in\mathbb{N}_{0}$
of even order derivatives of components of $y(t)$ . Namely,
\begin{equation}
\widehat{J}_{E}=y_{z_{1}}^{(2c_{1})}(t)y_{z_{2}}^{(2c_{2})}(t)\ldots y_{z_{q}}^{(2c_{q})}(t),\qquad c_{1},c_{2},\ldots,c_{q},q\in\mathbb{N}_{0}.\label{eq:ANYNUMBEREVENDERIVATIVES-1}
\end{equation}
Evidently, the induction hypothesis holds for the derivatives $y^{(1)}(0)=y^{(3)}(0)=y^{(5)}(0)=\overrightarrow{0}$.
It is easily verified that (\ref{eq:3RDERyj}) and (\ref{eq:5THDERyj})
are sums of products of the form (\ref{eq:ODDmDERIVATIVE}). The crux
of the induction is to show that $\widehat{J}_{O}$ and $\widehat{J}_{E}$
respectively, are of the desired form (\ref{eq:ODDDERIVATIVES=000026ODDNUMhat})
and (\ref{eq:ANYNUMBEREVENDERIVATIVES-1}) for any $m$ odd. To this
end differentiate twice both sides of (\ref{eq:ODDmDERIVATIVE}) .
This leads to the relation $\widehat{T_{m}^{(2)}}=Q_{1}+Q_{2}+Q_{3}$
\,where
\[
Q_{1}:=\sum_{k_{l+2}=1}^{n}\sum_{k_{l+1}=1}^{n}\{[\frac{\partial^{l+2}f_{j}(y(t))}{\partial y_{k_{l+2}}\partial y_{k_{l+1}}\partial y_{k_{l}}\ldots\partial y_{k_{2}}\partial y_{k_{1}}}]\}y_{k_{l+2}}^{(1)}y_{k_{l+1}}^{(1)}J_{O}J_{E}
\]
\begin{equation}
+\sum_{k_{l+1}=1}^{n}[\frac{\partial^{l+1}f_{j}(y(t))}{\partial y_{k_{l+1}}\partial y_{k_{l}}\ldots\partial y_{k_{2}}\partial y_{k_{1}}}]y_{k_{l+1}}^{(2)}J_{O}J_{E},\label{eq:INDUCTIVE1PART1}
\end{equation}
\begin{equation}
Q_{2}:=2\sum_{k_{l+1}=1}^{n}[\frac{\partial^{l+1}f_{j}(\overrightarrow{y}(t))}{\partial y_{k_{l+1}}\partial y_{k_{l}}\ldots\partial y_{k_{2}}\partial y_{k_{1}}}]y_{k_{l+1}}^{(1)}[J_{O}^{(1)}J_{E}+J_{O}J_{E}^{(1)}],\label{eq:INDUCTIVE1PART2}
\end{equation}
\begin{equation}
Q_{3}:=\frac{\partial^{l}f_{j}(y(t))}{\partial y_{k_{l}}\ldots\partial y_{k_{2}}\partial y_{k_{1}}}[J_{O}J_{E}]^{(2)}=\frac{\partial^{l}f_{j}(y(t))}{\partial y_{k_{l}}\ldots\partial y_{k_{2}}\partial y_{k_{1}}}[J_{O}^{(2)}J_{E}+2J_{O}^{(1)}J_{E}^{(1)}+J_{O}J_{E}^{(2)}].\label{eq:INDUCTIVE1PART3}
\end{equation}
We view $Q_{1},Q_{2}$ and $Q_{3}$ as sums of certain products and
we show that they are of the form (\ref{eq:ODDmDERIVATIVE+2}) subject
to (\ref{eq:ODDDERIVATIVES=000026ODDNUMhat}) and (\ref{eq:ANYNUMBEREVENDERIVATIVES-1}).
It is useful to keep in mind that $J_{O}$ is a product of an odd
number $r\geq1$ of odd order derivatives. The list of these terms
is:
\begin{equation}
y_{k_{l+2}}^{(1)}y_{k_{l+1}}^{(1)}J_{O}J_{E},\:y_{k_{l+1}}^{(2)}J_{O}J_{E},\;y_{k_{l+1}}^{(1)}J_{O}^{(1)}J_{E},\;y_{k_{l+1}}^{(1)}J_{O}J_{E}^{(1)},\;J_{O}^{(1)}J_{E}^{(1)},\;J_{O}^{(2)}J_{E},\;J_{O}J_{E}^{(2)}.\label{eq:ALLTERMS4INDUNTIONm+2}
\end{equation}
For the products originating from $y_{k_{l+2}}^{(1)}y_{k_{l+1}}^{(1)}J_{O}J_{E}$
put $\widehat{J}_{O}=y_{k_{l+2}}^{(1)}y_{k_{l+1}}^{(1)}J_{O},\:\widehat{J}=J_{E}.$
For the products originating from $y_{k_{l+1}}^{(2)}J_{O}J_{E}$ put
$\widehat{J}_{O}=J_{O},\:\widehat{J}_{O}=y_{k_{l+1}}^{(2)}J_{E}.$
Consider now $y_{k_{l+1}}^{(1)}J_{O}^{(1)}J_{E}$. Evidently, $J_{O}^{(1)}$
is a sum of $r$ products as follows. 
\begin{equation}
J_{O}^{(1)}=\sum_{j=1}^{r}y_{s_{j}}^{(2e_{j}+2)}\prod_{l\neq j,l=1}^{r}y_{s_{l}}^{(2e_{l}+1)}.\label{eq:JO1rstDERI}
\end{equation}
Each product in (\ref{eq:JO1rstDERI}) has $(r-1)$ odd order derivatives
factors and precisely one factor that is an even order derivative
of a certain component of $y_{s_{j}}$. For a representative product
in $y_{k_{l+1}}^{(1)}J_{O}^{(1)}J_{E}$ put

\begin{equation}
\widehat{J}_{E}=y_{s_{j}}^{(2e_{j}+2)}J_{E},\:r=1\Rightarrow\widehat{J}_{O}=y_{k_{l+1}}^{(1)},\;r=2\Rightarrow\widehat{J}_{O}=y_{k_{l+1}}^{(1)}\prod_{l\neq j,l=1}^{r}y_{s_{l}}^{(2e_{l}+1)}.\label{eq:JOHATJEHAT}
\end{equation}
Consider now a product originating from $y_{k_{l+1}}^{(1)}J_{O}J_{E}^{(1)}$.
If $w=0$ namely $J_{E}\equiv1$ , then $J_{E}^{(1)}\equiv0$ and
the product $y_{k_{l+1}}^{(1)}J_{O}J_{E}^{(1)}\equiv0$. If $w=1,$
put $\widehat{J}_{O}=y_{k_{l+1}}^{(1)}y_{z_{1}}^{(2c_{1}+1)}J_{O},\:\widehat{J}_{E}\equiv1$.
Let $w\geq2$. Then, $J_{E}^{(1)}$ is the following sum of $w$ products.

\begin{equation}
J_{E}^{(1)}=\sum_{j=1}^{w}y_{z_{j}}^{(2c_{j}+1)}\prod_{l\neq j,l=1}^{w}y_{z_{l}}^{(2c_{l})}.\label{eq:JE1rstDER}
\end{equation}
Put,

\begin{equation}
\widehat{J}_{O}=y_{k_{l+1}}^{(1)}y_{z_{j}}^{(2c_{j}+1)}J_{O},\;\widehat{J}_{E}=\prod_{l\neq j,l=1}^{w}y_{z_{l}}^{(2c_{l})}.\label{eq:2NDTERMDERGREEKJODERJE}
\end{equation}
Consider now the products emanating from $J_{O}^{(1)}J_{E}^{(1)}$.
If $J_{E}\equiv1$ namely $w=0$ , then $J_{O}^{(1)}J_{E}^{(1)}\equiv0$.
If $r=1$ and $w=1$ put $\widehat{J}_{O}=y_{z_{1}}^{(2c_{1}+1)},\;\widehat{J}_{E}=y_{s_{1}}^{(2e_{1}+2)}$.
If $r\geq2$ and $w=1$ put

\begin{equation}
\widehat{J}_{O}=y_{z_{1}}^{(2c_{1}+1)}\prod_{l\neq j,l=1}^{r}y_{s_{l}}^{(2e_{l}+1)},\;\widehat{J}_{E}\equiv1.\label{eq:HATSDERJOr>1DERJEw=00003D1}
\end{equation}
If $r\geq2$ and $w\geq2$ we have by (\ref{eq:JO1rstDERI}) and (\ref{eq:JE1rstDER})
$rw$ products in $J_{O}^{(1)}J_{E}^{(1)}$ . Put 
\[
\widehat{J}_{O}=y_{z_{j}}^{(2c_{j}+1)}[\prod_{l\neq j,l=1}^{r}y_{s_{l}}^{(2e_{l}+1)}],\;\widehat{J}_{E}=y_{s_{j}}^{(2e_{j}+2)}[\prod_{l\neq j,l=1}^{w}y_{z_{l}}^{(2c_{l})}].
\]
It remains to consider the factors originating from $J_{O}^{(2)}J_{E}$
and from $J_{O}J_{E}^{(2)}$. To this end we first calculate the sum
of products emanating from $J_{O}^{(2)}$ and multiply them by $J_{E}$.
If $J_{O}$ has only $r=1$ factors then $J_{O}^{(2)}=y_{s_{1}}^{(2e_{1}+3)}$
. Then put $\widehat{J}_{O}=y_{s_{1}}^{(2e_{1}+3)}$ and $\widehat{J}_{E}=J_{E}.$
If $r\geq3$ then the factors in $J_{O}^{(2)}$ are of \,two kinds.
The one is
\begin{equation}
y_{s_{j}}^{(2e_{j}+2)}y_{s_{k}}^{(2e_{k}+2)}\prod_{l\neq j,k,l=1}^{r}y_{s_{l}}^{(2e_{l}+1)},r\geq3.\label{eq:JODER2products1}
\end{equation}
Then, put 
\[
\widehat{J}_{O}=\prod_{l\neq j,k,l=1}^{r}y_{s_{l}}^{(2e_{l}+1)},\quad\widehat{J}_{E}:=y_{s_{j}}^{(2e_{j}+2)}y_{s_{k}}^{(2e_{k}+2)}J_{E},r\geq3.
\]
The other type of product in $J_{O}^{(2)}$ is $y_{s_{j}}^{(2e_{j}+3)}\prod_{l\neq j,l=1}^{r}y_{s_{l}}^{(2e_{l}+1)}.$
Put
\[
\widehat{J}_{O}:=y_{s_{j}}^{(2e_{j}+3)}\prod_{l\neq j,l=1}^{r}y_{s_{l}}^{(2e_{l}+1)},\quad\widehat{J}_{E}:=J_{E},\quad if\qquad r\geq3.
\]
It remains to analyze the resulting products in $J_{O}J_{E}^{(2)}$.
To this end calculate first the resulting products in $J_{E}^{(2)}$
. If $w=0$ then $J_{E}^{(2)}\equiv0$ and then the contribution of
products from $J_{O}J_{E}^{(2)}$ is $0$. If $w=1$ then $J_{E}^{(2)}\equiv y_{z_{1}}^{(2c_{1}+2)}$
and we put $\widehat{J}_{O}=J_{O},\quad\widehat{J}_{E}=y_{z_{1}}^{(2c_{1}+2)}$.
Assume that $w\geq2$. Then $J_{E}^{(2)}$ is a sum of products that
are of two kinds. The first kind is

\begin{equation}
y_{z_{j}}^{(2c_{1}+2)}\prod_{l\neq j,l=1}^{w}y_{z_{l}}^{(2c_{l})}.\label{eq:JOJODer2}
\end{equation}
Then put 
\begin{equation}
\widehat{J}_{O}=J_{O},\;\widehat{J}_{E}=y_{z_{j}}^{(2c_{1}+2)}\prod_{l\neq j,l=1}^{w}y_{z_{l}}^{(2c_{l})}.\label{eq:JODer2JEALPHA=00003D2-1}
\end{equation}
The second kind is precisely $y_{z_{j}}^{(2c_{j}+1)}y_{z_{k}}^{(2c_{k}+1)}$
if $w=2.$ Then put

\begin{equation}
\widehat{J}_{O}\equiv y_{z_{j}}^{(2c_{j}+1)}y_{z_{k}}^{(2c_{k}+1)}J_{O},\quad\widehat{J}_{E}\equiv1.\label{eq:JEAlphaj,k=00003D1,1w=00003D2-2}
\end{equation}
If $w\geq3$ then the second kind of a product resulting from the
sum of products in $J_{E}^{(2)}$ is
\begin{equation}
y_{z_{j}}^{(2c_{j}+1)}y_{z_{k}}^{(2c_{k}+1)}\prod_{l\neq j,k,l=1}^{w}y_{z_{l}}^{(2c_{l})}.\label{eq:JOAlphaj=00003D1Alphak=00003D1-2-1}
\end{equation}
Then put

\begin{equation}
\widehat{J}_{O}:=y_{z_{j}}^{(2c_{j}+1)}y_{z_{k}}^{(2c_{k}+1)}J_{O},\quad\widehat{J}_{E}:=\prod_{l\neq j,k,l=1}^{w}y_{z_{l}}^{(2c_{l})}.\label{eq:JEAlphaj,k=00003D1andw>2-2}
\end{equation}
In sum, all $T_{m}(0)=0$ imply that all $\widehat{T_{m}^{(2)}}(0)=0$.
Consequently, $y^{(m)}(0)=\overrightarrow{0}$ for odd $m$$\in\mathbb{N}_{0}$
. Q.E.D.
\end{proof}
\begin{rem*}
It is evident that the symmetry $y(t)\equiv y(-t)$ is manifest to
the future being a perfect reflection of the past. We show inhere
that all odd order derivatives, in the Taylor series expansion of
$y(t)$, vanish at $t=0$. The estimate $\left|t\right|\leq\sqrt{\nicefrac{2b}{M}}$
follows by adaptations of the technicians in e.g. \cite{P.F.Hsieh=000026Y. Sibuya-2}
, Chapter 1, pages 20. Compare also with \cite{HILLEODESCOMPLEX-1}.
The Taylor series then show that $y(t)=y(0)+\sum_{l=1}^{n}\nicefrac{[y^{(2l)}(0)}{(2l)!}]t^{2l}$
satisfy $y(t)\equiv y(-t)$. Since $y''=f(y)$ is an autonomous system,
then for any $t_{0}\in\mathbb{C}$, $y(t)=y(t_{0})+\sum_{l=1}^{n}[\nicefrac{y^{(2l)}(t_{0})}{(2l)!}](t-t_{0})^{2l}$
are also solutions of $y''=f(y)$. Furthermore, the velocities are
symmetric functions as well. Namely $-y'(t)\equiv y'(-t)$. The future
accelerations and forces acting on the bodies, are a perfect reflection
of their past. Let $t_{c}\in\mathbb{\mathbb{R}}$ be a real valued
collision time, where $y_{k}(t_{c})=y_{j}(t_{c})$ for some $\;k\neq j$.
Allowing the variable $t$ to be complex valued could make it possible
to analytically continue a solution $y(t)$ from the real line into
the complex plane, from time $t<t_{c}$ to $t>t_{c}$. Then, $y(t)\equiv y(-t)$
holds for $t<t_{c}$ as well as for $t>t_{c}$ . Thus, circumventing
a collision at time $t=t_{c}$.
\end{rem*}

\section{a lemma on partial derivatives of $f(y)$}

We clarify now what is an even and an odd function of a scalar function
$u=$$H(y_{1},y_{2},\ldots,y_{N})$ of several variables. To this
end we denote the transposed column vector $y^{T}=(y_{1},y_{2},\ldots,y_{N})$
and we put $u=$$H(y_{1},y_{2},\ldots,y_{N})=H(y)$.
\begin{defn}
Denote by $REG$ an open connected set in $\mathbb{R}^{N}$ . We say
that $H(y)$ is an even function of $y$ in $REG$ if
\begin{equation}
H(y)=H(-y),\:y\in REG.\label{eq:EVEN(VECTORy)}
\end{equation}
We say that $H(y)$ is an odd function of $y$ in $REG$ if
\end{defn}

\begin{equation}
-H(y)=H(-y),\:y\in REG.\label{eq:ODD(VECTORy)}
\end{equation}
This definition is different than requiring that $u=$$H(y_{1},y_{2},\ldots,y_{N})=H(y)$
be an even or an odd function in each individual coordinate $y_{j}$
. In order to bring out the difference we add the following.
\begin{defn}
Denote by $REG$ an open connected set in $\mathbb{R}^{N}$ . We say
that $H(y)$ is an even function of $y$ in $REG$ in the strict sense
if
\begin{equation}
H(y_{1},y_{2},\ldots,y_{j},\ldots,y_{N})=H(y_{1},y_{2},\ldots,-y_{j},\ldots,y_{N}),\:y\in REG,\:j=1,2,\ldots,N.\label{eq:EVEN(VECTORy)-1}
\end{equation}
We say that $H(y)$ is an odd function of $y$ in $REG$ in the stricter
sense if

\begin{equation}
-H(y_{1},y_{2},\ldots,y_{j},\ldots,y_{N})=H(y_{1},y_{2},\ldots,-y_{j},\ldots,y_{N}),\:y\in REG,\:j=1,2,\ldots,N.\label{eq:ODD(VECTORy)-1}
\end{equation}
Consider the following functions 
\begin{equation}
H(y_{1},y_{2}):=y_{1}^{5}y_{2}^{3},\:L(y_{1},y_{2})=y_{1}^{10}y_{2}^{6}.\label{eq:EXAMPLE1MULTIVARIABLES}
\end{equation}
 Evidently, $H(y_{1},y_{2}):=y_{1}^{5}y_{2}^{3}$ is an even function
in $REG:=\mathbb{R}^{2}$ . However, it is an odd function in the
strict sense in $REG:=\mathbb{R}^{2}$. Evidently, $L(y_{1},y_{2})=y_{1}^{10}y_{2}^{6}$
is an even function in $REG:=\mathbb{R}^{2}$ and it is also an even
function in the strict sense in $REG:=\mathbb{R}^{2}$.
\end{defn}

\begin{rem}
The reader may want to consider a multinomial in the $(r+w)$ independent
variables $y_{1},y_{2},$$\ldots,y_{j},\ldots,y_{r},y_{r+1},y_{r+2},\ldots,y_{r+w}$.
\begin{equation}
H(y)=[y_{1}^{[2e_{1}+1]}y_{2}^{[2e_{2}+1]}\ldots y_{j}^{[2e_{j}+1]}\ldots y_{r}^{[2e_{r}+1]}][y_{r+1}^{[2c_{1}]}y_{r+2}^{[2c_{2}]}\ldots y_{r+w}^{[2c_{w}]}]\label{eq:MULTINOMIAL}
\end{equation}
where $e_{1},e_{2},\ldots e_{r},c_{1},c_{2},\ldots,c_{w}$$\in\mathbb{N}_{0},\;,r,w\in\mathbb{N}\;$
. Formulation of necessary and sufficient conditions on the powers
occurring in $H(y)$ such that a) $H(y)$ is an even multivariate
function b) $H(y)$ is an even multivariate function in the strict
sense c) $H(y)$ is an odd multivariate function d) $H(y)$ is an
odd multivariate function in the strict sense, could further clarify
the difference between these two types of symmetry. Next we formulate
an analog to Lemma 1 for multivariate functions.
\end{rem}

\begin{lem}
Let $H(y)\in C^{1}(REG)$. i) Assume that $H(y)=H(-y),$$\:y\in REG$.
Then the partial derivatives
\[
\varPsi_{j}(y):=\frac{\partial H(y)}{\partial y_{j}},\:j=1,2,\cdots,N,
\]
are odd function in $REG$. ii) Assume that $-H(y)=H(-y),$$\:y\in REG$.
Then the partial derivatives
\[
\varPsi_{j}(y):=\frac{\partial H(y)}{\partial y_{j}},\:j=1,2,\cdots,N,
\]
are even functions in $REG$. iii) Assume $f(y)$ to be a column vector
function, $f^{T}(y):=[f_{1}(y),f_{2}(y),\ldots,f_{n}(y)]$ , where
$f_{j}(y),\:j=1,2,\ldots,n$ are the scalar component of $f(y)$ such
that $f_{j}(y)\in C^{1}(REG)$. Then,
\begin{equation}
f(y)=f(-y),\:y\in REG\Longrightarrow\varPsi(y):=\frac{\partial f(y)}{\partial y_{j}}=-\varPsi(-y):=-\frac{\partial f(-y)}{\partial y_{j}},\:j=1,2,\ldots,N.\label{eq:DERVECTORFUNCAREODDif=00005Bf=00005Deven}
\end{equation}
Moreover,

\begin{equation}
-f(y)=f(-y),\:y\in REG\Longrightarrow\varPsi(y):=\frac{\partial f(y)}{\partial y_{j}}=\varPsi(-y):=\frac{\partial f(-y)}{\partial y_{j}},\:j=1,2,\ldots,N.\label{eq:DERVECFUNCTIONSAREEVENIF=00005Bf=00005DisODD}
\end{equation}
We first prove i) and we focus on the quotient below: with $y\in REG$
; with $h\neq0$ ; and with $h$ arbitrarily small.
\begin{equation}
Q_{j}(y,h):=\frac{H(y_{1},y_{2},\ldots,y_{j-1},y_{j}+h,y_{j+1},\ldots,y_{N})-H(y_{1},y_{2},\ldots,y_{j-1},y_{j},y_{j+1},\ldots,y_{N})}{h}.\label{eq:Quotient1-2}
\end{equation}
Put as short hand notation
\begin{equation}
\widehat{H}(y_{j}+h):=H(y_{1},y_{2},\ldots,y_{j-1},y_{j}+h,y_{j+1},\ldots,y_{N}),\label{eq:H(y_j+h)Abbreviation}
\end{equation}
\begin{equation}
\widehat{H}(-y_{j}-h):=H(-y_{1},-y_{2},\ldots,-y_{j-1},-y_{j}-h,-y_{j+1},\ldots,-y_{N}).\label{eq:H(y_j+h)Abbreviation-1}
\end{equation}
\end{lem}

\begin{proof}
Since $H(y)$ is an even function then 
\begin{equation}
H(y_{1},y_{2},\ldots,y_{j-1},y_{j}+h,y_{j+1},\ldots,y_{N})=H(-y_{1},-y_{2},\ldots,-y_{j-1},-y_{j}-h,-y_{j+1},\ldots,-y_{N}).\label{eq:QUOTIENTCHNGEDBYSYMMETRY1-2}
\end{equation}
\begin{equation}
H(y_{1},y_{2},\ldots,y_{j-1},y_{j},y_{j+1},\ldots,y_{N})=H(-y_{1},-y_{2},\ldots,-y_{j-1},-y_{j},-y_{j+1},\ldots,-y_{N}).\label{eq:QUOTIENTCHNGEDBYSYMMETRY1-2 1}
\end{equation}
Substitute from (\ref{eq:QUOTIENTCHNGEDBYSYMMETRY1-2}) and (\ref{eq:H(y_j+h)Abbreviation})
and (\ref{eq:H(y_j+h)Abbreviation-1}) into the right hand side of
(\ref{eq:Quotient1-2}) to obtain 
\begin{equation}
Q_{j}(y,h)=-Q_{j}(-y,-h)=-\frac{\widehat{H}(-y_{j}-h)-\widehat{H}(-y_{j})}{-h}.\label{eq:RELATIONQ=00003D-Q-2}
\end{equation}
 Take the limit as $h\rightarrow0$ in (\ref{eq:RELATIONQ=00003D-Q-2})
and obtain

\begin{equation}
\varPsi(y):=\frac{\partial H(y)}{\partial y_{j}}=lim_{h\rightarrow0}Q_{j}(y,h)=-lim_{h\rightarrow0}Q_{j}(-y,-h)=-\varPsi(-y).
\end{equation}
Next we prove ii). We focus again on the quotient 
\begin{equation}
Q_{j}(y,h):=\frac{H(y_{1},y_{2},\ldots,y_{j-1},y_{j}+h,y_{j+1},\ldots,y_{N})-H(y_{1},y_{2},\ldots,y_{j-1},y_{j},y_{j+1},\ldots,y_{N})}{h}.\label{eq:Quotient1-1-1}
\end{equation}
Since $H(y)$ is an odd function then

\[
H(y_{1},y_{2},\ldots,y_{j-1},y_{j}+h,y_{j+1},\ldots,y_{N})
\]
\begin{equation}
=-H(-y_{1},-y_{2},\ldots,-y_{j-1},-y_{j}-h,-y_{j+1},\ldots,-y_{N}).\label{eq:QUOTIENTCHNGEDBYSYMMETRY1-2-1}
\end{equation}
\begin{equation}
H(y_{1},y_{2},\ldots,y_{j-1},y_{j},y_{j+1},\ldots,y_{N})=-H(-y_{1},-y_{2},\ldots,-y_{j-1},-y_{j},-y_{j+1},\ldots,-y_{N}).\label{eq:QUOTIENTCHNGEDBYSYMMETRY1-2-1 1}
\end{equation}
Substitute from (\ref{eq:QUOTIENTCHNGEDBYSYMMETRY1-2-1}) into the
right hand side of (\ref{eq:Quotient1-1-1}) to obtain 
\begin{equation}
Q_{j}(y,h)=Q_{j}(-y,-h).\label{eq:RELATIONQ=00003D-Q-1-1}
\end{equation}
Take the limit as $h\rightarrow0$ in (\ref{eq:RELATIONQ=00003D-Q-1-1})
and obtain

\begin{equation}
\varPsi(y):=\frac{\partial H(y)}{\partial y_{j}}=lim_{h\rightarrow0}Q_{j}(y,h)=lim_{h\rightarrow0}Q_{j}(-y,-h)=\varPsi(-y).\label{eq:ODDDERH}
\end{equation}
The proof of iii) follows from the proofs of i) and ii) and the definition
of $f(y)$.
\end{proof}

\section{sample of two successive even derivatives lemma}

A straight forward calculation reveals that
\begin{equation}
\frac{d^{4}y_{j}(t)}{dt^{4}}=\sum_{k_{1}=1}^{n}\frac{\partial f_{j}(y(t))}{\partial y_{k_{1}}}\frac{d^{2}y_{k_{1}}(t)}{dt^{2}}+\sum_{k_{1}=1}^{n}\sum_{k_{2}=1}^{n}\frac{\partial^{2}f_{j}(y(t))}{\partial y_{k_{2}}\partial y_{k_{1}}}\frac{dy_{k_{2}}(t)}{dt}\frac{dy_{k_{1}}(t)}{dt}.\label{eq:4THDERSummar}
\end{equation}

{\small{}
\[
\frac{d^{6}y_{j}(t)}{dt^{6}}=\sum_{k_{1}=1}^{n}\frac{\partial f_{j}(y(t))}{\partial y_{k_{1}}}\frac{d^{4}y_{k_{1}}(t)}{dt^{4}}
\]
\[
+\sum_{k_{1}=1}^{n}\sum_{k_{2}=1}^{n}\frac{\partial^{2}f_{j}(y(t))}{\partial y_{k_{2}}\partial y_{k_{1}}}[3\frac{d^{3}y_{k_{1}}(t)}{dt^{3}}\frac{dy_{k_{2}}(t)}{dt}+\frac{dy_{k_{1}}(t)}{dt}\frac{d^{3}y_{k_{2}}(t)}{dt^{3}}+3\frac{d^{2}y_{k_{1}}(t)}{dt^{2}}\frac{d^{2}y_{k_{2}}(t)}{dt^{2}}]
\]
\[
+\sum_{k_{1}=1}^{n}\sum_{k_{2}=1}^{n}\sum_{k_{3}=1}^{n}\frac{\partial^{3}f_{j}(y(t))}{\partial y_{k_{3}}\partial y_{k_{2}}\partial y_{k_{1}}}[3\frac{d^{2}y_{k_{1}}(t)}{dt^{2}}\frac{dy_{k_{2}}(t)}{dt}\frac{dy_{k_{3}}(t)}{dt}+2\frac{dy_{k_{1}}(t)}{dt}\frac{d^{2}y_{k_{2}}(t)}{dt^{2}}\frac{dy_{k_{3}}(t)}{dt}+\frac{dy_{k_{1}}(t)}{dt}\frac{dy_{k_{2}}(t)}{dt}\frac{d^{2}y_{k_{3}}(t)}{dt^{2}}]
\]
\begin{equation}
+\sum_{k_{1}=1}^{n}\sum_{k_{2}=1}^{n}\sum_{k_{3}=1}^{n}\sum_{k_{4}=1}^{n}\frac{\partial^{4}f_{j}(y(t))}{\partial y_{k_{4}}\partial y_{k_{3}}\partial y_{k_{2}}\partial y_{k_{1}}}\frac{dy_{k_{4}}(t)}{dt}\frac{dy_{k_{3}}(t)}{dt}\frac{dy_{k_{2}}(t)}{dt}\frac{dy_{k_{1}}(t)}{dt}.\label{eq:6THDeryj(t=00003D0)Summary}
\end{equation}
}These aid us in formulating the elements of induction i) , ii) in
theorem 6 below.

\section{induction on even derivatives}

We prove 
\begin{thm}
Assume that: i) $t\in\mathbb{C},\:t_{0}\in\mathbb{R},y_{0},\eta\in\mathbb{\mathbb{\mathbb{R}}}^{3N},y\,,f(y)\in\mathbb{C}^{n},\;n\in\mathbb{N}$,
where $f(y)$ is an analytic function in the vector variable $y$
in a disk such that in $D$
\begin{equation}
D:=\{y\bigl|\left\Vert y-y_{0}\right\Vert \leq b\}\Longrightarrow\left\Vert f(y)\right\Vert \leq M.\label{eq:Boundedf(y)in a disk.-1}
\end{equation}
 Then, the initial value problem

\begin{equation}
y''=f(y),\;y(t_{0})=\overrightarrow{0},y'(t_{0})=\eta,\;\overrightarrow{0}^{T}:=[0,0,\cdots,0],\label{eq:INITIALVALUEPROBGENERAL-1}
\end{equation}
possesses a unique analytic solution $y(t)$ for $\left|t-t_{0}\right|\leq\sqrt{\nicefrac{2b}{M}}$
that satisfies $-y(t-t_{0})\equiv y(-(t-t_{0}))$. Namely, $y^{(m)}(t_{0})=\overrightarrow{0}$
for all even numbers m. 
\end{thm}

\begin{proof}
We proceed to prove this theorem by induction. The proof uses some
calculations and representations analogous to those in section 3.
However, these come with different interpretations. This is necessitated
by conditions \, i), ii) of theorem 6, that are different than the
analogous conditions \, i) and ii) in theorem 1. In theorem 6, the
assumption $-f(y)=f(-y)$ is the source of the differences. The following
claim clarifies the role of i) in this theorem 6.{\small{} If $f(y)$
is an odd function of $y$ then $f(\overrightarrow{0})=\overrightarrow{0}$.
Moreover, by lemma 5, the even order partial derivatives of $f_{j}(y)$
with respect to the variables $y_{k}$, ( like $f_{j}^{(0)}(y):=f(y)$
) , are odd functions of $y$ . Namely,
\begin{equation}
-\frac{\partial^{l}f_{j}(y(t))}{\partial y_{k_{l}}\ldots\partial y_{k_{2}}\partial y_{k_{1}}}=\frac{\partial^{l}f_{j}(-y(t))}{\partial y_{k_{l}}\ldots\partial y_{k_{2}}\partial y_{k_{1}}},\;l=0,2,4,\ldots,\Longrightarrow\frac{\partial^{l}f_{j}(\overrightarrow{0})}{\partial y_{k_{l}}\ldots\partial y_{k_{2}}\partial y_{k_{1}}}=0.\label{eq:ODDORDERDERfjyEVENORDERDERFjyl=00003D0-1}
\end{equation}
Per lemma 5, the odd order partial derivatives of $f_{j}(y)$ with
respect to $y_{j}$ (unlike $f_{j}^{(0)}(y)$ $=f(y)$) are even functions
of $y$ . Namely,}{\small\par}

\[
\frac{\partial^{l}f_{j}(y(t))}{\partial y_{k_{l}}\ldots\partial y_{k_{2}}\partial y_{k_{1}}}=\frac{\partial^{l}f_{j}(-y(t))}{\partial y_{k_{l}}\ldots\partial y_{k_{2}}\partial y_{k_{1}}},\;l\;is\;odd.
\]

Without loss of generality assume that $t_{0}=0$ since our differential
system is autonomous. Assume that each component $y_{j}^{(m)}$, $j=1,2,\cdots,n$
, is a finite sum of terms of the form
\begin{equation}
T_{m}:=\frac{\partial^{l}f_{j}(y(t))}{\partial y_{k_{l}}\ldots\partial y_{k_{2}}\partial y_{k_{1}}}J_{O}J_{E},\label{eq:EVENDlGREEKJOJE4}
\end{equation}
where: i) $J_{E}$ is a finite number $w$ of \,factors of even order
derivatives of the form $y_{j}^{(2c)}$ . Namely,

\begin{equation}
J_{E}=y_{z_{1}}^{(2c_{1})}(t)y_{z_{2}}^{(2c_{2})}(t)\ldots y_{z_{w}}^{(2c_{w})}(t),\qquad c_{1},c_{2},\ldots,c_{w}\in\mathbb{N}_{0}.\label{eq:ANYNUMBEREVENDERIVATIVES-2}
\end{equation}
If $l\thinspace$is an odd number then $w\geq1$ is an odd number.
If $l\thinspace$is an even number then $w\geq0$ is an even number.
Thus, making $T_{m}(0)=0$ in (\ref{eq:EVENDlGREEKJOJE4}) and consequently
$y^{(m)}(0)=0.$ ii) $J_{O}$ has an even number $r$ of\, factors
of the form $y_{j}^{(2g+1)}$. Namely,

\begin{equation}
J_{O}=y_{s_{1}}^{(2e_{`1}+1)}(t)y_{s_{2}}^{(2e_{2}+1)}(t)\ldots y_{s_{r}}^{(2e_{r}+1)}(t),\qquad e_{1},e_{2},\ldots,e_{r}\in\mathbb{N}_{0}.\label{eq:ODDDERIVATIVES=000026ODDNUMBER-1}
\end{equation}
Then, $y_{j}^{(m+2)}(t)$ is a finite sum of terms of the form
\begin{equation}
\widehat{T_{m}^{(2)}}:=\frac{\partial^{s}f_{j}(y(t))}{\partial y_{k_{s}}\ldots\partial y_{k_{2}}\partial y_{k_{1}}}\widehat{J}_{O}\widehat{J}_{E},\label{eq:EVENrDERIVATIVE+2-1}
\end{equation}
where: i) $\widehat{J}_{E}$ is a finite product of a number $q$
of \,even order derivatives of components of $y(t)$ . Namely,
\begin{equation}
\widehat{J}_{E}=y_{z_{1}}^{(2c_{1})}(t)y_{z_{2}}^{(2c_{2})}(t)\ldots y_{z_{q}}^{(2c_{q})}(t),\qquad c_{1},c_{2},\ldots,c_{q},q\in\mathbb{N}_{0}.\label{eq:ANYNUMBEREVENDERIVATIVESEven}
\end{equation}
If $s$ is an odd number then $w\geq1$ is an odd number. If $s$
is an even number then $w\geq0$ is an even number. Thus, making $\widehat{T_{m}^{(2)}}(0)=0$
and consequently $y^{(m+2)}(0)=0.$ ii) $\widehat{J}_{O}$ is a finite
product of an even number $p\in\mathbb{N}$ of odd order derivatives
of components of $y(t)$. Namely,
\begin{equation}
\widehat{J}_{O}=y_{s_{1}}^{(2g_{1}+1)}(t)y_{s_{2}}^{(2g_{+1})}(t)\ldots y_{s_{p}}^{(2g_{p}+1)}(t)\quad g_{1},g_{2},\ldots,g_{p}\in\mathbb{N}_{0}.\label{eq:ODDDERIVATIVES=000026ODDNUMhat-1}
\end{equation}

A calculation of the second derivative $T_{m}^{(2)}(t)$ shows that
\begin{equation}
T_{m}^{(2)}(t)=Q_{1}+Q_{2}+Q_{3},\label{eq:Q1+Q2+Q3EVEN}
\end{equation}
where

\[
Q_{1}:=\sum_{k_{l+2}=1}^{n}\sum_{k_{l+1}=1}^{n}\{[\frac{\partial^{l+2}f_{j}(y(t))}{\partial y_{k_{l+2}}\partial y_{k_{l+1}}\partial y_{k_{l}}\ldots\partial y_{k_{2}}\partial y_{k_{1}}}]\}y_{k_{l+2}}^{(1)}y_{k_{l+1}}^{(1)}J_{O}J_{E}
\]
\begin{equation}
+\sum_{k_{l+1}=1}^{n}[\frac{\partial^{l+1}f_{j}(y(t))}{\partial y_{k_{l+1}}\partial y_{k_{l}}\ldots\partial y_{k_{2}}\partial y_{k_{1}}}]y_{k_{l+1}}^{(2)}J_{O}J_{E},\label{eq:INDUCTIVE1PART1-1}
\end{equation}
\begin{equation}
Q_{2}:=2\sum_{k_{l+1}=1}^{n}[\frac{\partial^{l+1}f_{j}(y(t))}{\partial y_{k_{l+1}}\partial y_{k_{l}}\ldots\partial y_{k_{2}}\partial y_{k_{1}}}]y_{k_{l+1}}^{(1)}[J_{O}^{(1)}J_{E}+J_{O}J_{E}^{(1)}],\label{eq:INDUCTIVE1PART2-1}
\end{equation}
\begin{equation}
Q_{3}:=\frac{\partial^{l}f_{j}(y(t))}{\partial y_{k_{l}}\ldots\partial y_{k_{2}}\partial y_{k_{1}}}[J_{O}J_{E}]^{(2)}=\frac{\partial^{l}f_{j}(y(t))}{\partial y_{k_{l}}\ldots\partial y_{k_{2}}\partial y_{k_{1}}}[J_{O}^{(2)}J_{E}+2J_{O}^{(1)}J_{E}^{(1)}+J_{O}J_{E}^{(2)}].\label{eq:INDUCTIVE1PART3-1}
\end{equation}
Below is the list of the different types of products of the form $\widehat{T_{m}^{(2)}}$
that occur in (\ref{eq:Q1+Q2+Q3EVEN}) . The products are:
\begin{equation}
\{[\frac{\partial^{l+2}f_{j}(y(t))}{\partial y_{k_{l+2}}\partial y_{k_{l+1}}\partial y_{k_{l}}\ldots\partial y_{k_{2}}\partial y_{k_{1}}}]\}y_{k_{l+2}}^{(1)}y_{k_{l+1}}^{(1)}J_{O}J_{E};\quad[\frac{\partial^{l+1}f_{j}(y(t))}{\partial y_{k_{l+1}}\partial y_{k_{l}}\ldots\partial y_{k_{2}}\partial y_{k_{1}}}]y_{k_{l+1}}^{(2)}J_{O}J_{E};\label{eq:Tm2DERHATQ1}
\end{equation}
\begin{equation}
[\frac{\partial^{l+1}f_{j}(y(t))}{\partial y_{k_{l+1}}\partial y_{k_{l}}\ldots\partial y_{k_{2}}\partial y_{k_{1}}}]y_{k_{l+1}}^{(1)}J_{O}^{(1)}J_{E};\quad[\frac{\partial^{l+1}f_{j}(y(t))}{\partial y_{k_{l+1}}\partial y_{k_{l}}\ldots\partial y_{k_{2}}\partial y_{k_{1}}}]y_{k_{l+1}}^{(1)}J_{O}J_{E}^{(1)};\label{eq:Tm2DERHATQ2}
\end{equation}
\begin{equation}
[\frac{\partial^{l}f_{j}(y(t))}{\partial y_{k_{l}}\ldots\partial y_{k_{2}}\partial y_{k_{1}}}]J_{O}^{(2)}J_{E};\quad[\frac{\partial^{l}f_{j}(y(t))}{\partial y_{k_{l}}\ldots\partial y_{k_{2}}\partial y_{k_{1}}}]J_{O}^{(1)}J_{E}^{(1)};\quad[\frac{\partial^{l}f_{j}(y(t))}{\partial y_{k_{l}}\ldots\partial y_{k_{2}}\partial y_{k_{1}}}]J_{O}J_{E}^{(2)};\label{eq:Tm2HATQ3}
\end{equation}
 Now we proceed to show that each product in (\ref{eq:Q1+Q2+Q3EVEN})
is of the desired form . Consider first each term in (\ref{eq:Tm2DERHATQ1})
and start with 
\begin{equation}
Q_{11}:=[\frac{\partial^{l+2}f_{j}(y(t))}{\partial y_{k_{l+2}}\partial y_{k_{l+1}}\partial y_{k_{l}}\ldots\partial y_{k_{2}}\partial y_{k_{1}}}]\}y_{k_{l+2}}^{(1)}y_{k_{l+1}}^{(1)}J_{O}J_{E}.\label{eq:Q11}
\end{equation}
Put,
\begin{equation}
\widehat{J}_{O}:=y_{k_{l+2}}^{(1)}y_{k_{l+1}}^{(1)}J_{O},\quad\widehat{J}_{E}:=J_{E}.\label{eq:NEWODDEVENO1=000026O2-1}
\end{equation}
Observe that $J_{O}$ has an even number of factors of odd order derivatives
that is $r\geq0$ . Consequently, $\widehat{J}_{O}$ has an even number
of odd order derivatives of components of $y(t)$ that is $(r+2)$
as required by ii). Assume that $l\thinspace$ is an odd number then
$s=l+2$ is also an odd number. Assume that $l$ is an even number
then $s=l+2$ is also an even number. Since $\widehat{J}_{E}:=J_{E}$
then all conditions of i) are satisfied. And hence all conditions
in (\ref{eq:EVENrDERIVATIVE+2-1}) holds. Then we have as desired

\begin{equation}
[\frac{\partial^{l}f_{j}(y(0))}{\partial y_{k_{l+2}}\partial y_{k_{l+1}}\partial y_{k_{l}}\ldots\partial y_{k_{2}}\partial y_{k_{1}}}]J_{O}(0)J_{E}(0)=0\Longrightarrow[\frac{\partial^{l+2}f_{j}(y(0))}{\partial y_{k_{l+2}}\partial y_{k_{l+1}}\partial y_{k_{l}}\ldots\partial y_{k_{2}}\partial y_{k_{1}}}]\widehat{J}_{O}(0)\widehat{J}_{E}(0)=0.\label{eq:ZERO101}
\end{equation}
Next we focus on the second representative product in (\ref{eq:Tm2DERHATQ1})
that is 
\begin{equation}
Q_{21}:=[\frac{\partial^{l+1}f_{j}(y(t))}{\partial y_{k_{l+1}}\partial y_{k_{l}}\ldots\partial y_{k_{2}}\partial y_{k_{1}}}]y_{k_{l+1}}^{(2)}J_{O}J_{E}.\label{eq:Q21}
\end{equation}
 Put 
\begin{equation}
\widehat{J}_{O}:=J_{O},\quad\widehat{J}_{E}:=y_{k_{l+1}}^{(2)}J_{E}.\label{eq:2NDTERMINDERAGAINAGAIN-1}
\end{equation}
If $w$ is the number of factors of even order derivatives in $J_{E}$
then the number of even order derivatives in $\;\widehat{J}_{E}$
is $w+1$. Observe that $s=l+1$. Thus condition i) in (\ref{eq:EVENrDERIVATIVE+2-1})
holds. Next we put 
\begin{equation}
Q_{12}:=[\frac{\partial^{l+1}f_{j}(y(t))}{\partial y_{k_{l+1}}\partial y_{k_{l}}\ldots\partial y_{k_{2}}\partial y_{k_{1}}}]y_{k_{l+1}}^{(1)}J_{O}^{(1)}J_{E}.\label{eq:Q12}
\end{equation}
There is an even number $r$ of factors of odd derivatives in $J_{O}$
that make $J_{O}^{(1)}$ a sum of $r$ products as follows.

\begin{equation}
J_{O}^{(1)}=y_{s_{1}}^{(2e_{1}+2)}y_{s_{2}}^{(2e_{2}+1)}\ldots y_{s_{r}}^{(2e_{r}+1)}+y_{s_{1}}^{(2e_{1}+1)}y_{s_{2}}^{(2e_{2}+2)}\ldots y_{s_{r}}^{(2e_{r}+1)}+\ldots+y_{s_{1}}^{(2e_{1}+1)}y_{s_{2}}^{(2e_{2}+2)}\ldots y_{s_{r}}^{(2e_{r}+2)}.\label{eq:JODERIVATive-1}
\end{equation}
Each summand in (\ref{eq:JODERIVATive-1}) is a product of $(r-1)$
odd order derivatives and precisely one factor is an even order derivative
of a certain component of $y$. Without loss of generality we relabel
each product in (\ref{eq:JODERIVATive-1}) as 
\begin{equation}
J_{OS}^{(1)}:=y_{s_{1}}^{(2u_{1}+1)}y_{s_{2}}^{(2u_{2}+1)}\ldots y_{s_{r-1}}^{(2u_{r-1}+1)}y_{s_{r}}^{(2u_{r}+2)},\label{eq:JOSDER1-1}
\end{equation}
and put 
\begin{equation}
\widehat{J}_{O}=y_{k_{l+1}}^{(1)}y_{s_{1}}^{(2u_{1}+1)}y_{s_{2}}^{(2u_{2}+1)}\ldots y_{s_{r-1}}^{(2u_{r-1}+1)},\:\widehat{J}_{E}=y_{s_{r}}^{(2u_{r}+2)}J_{E}.\label{eq:JOHATJEHAT-1}
\end{equation}
Evidently, $\widehat{J}_{O}$ , like $J_{O}$ , has the same even
number $r$ of factors of odd order derivatives of components of $y$
. Observe that $\widehat{J}_{E}$ has $(w+1)$ number of factors of
even order derivatives that is one more than $J_{E}$. However, $s=l+1$.
Therefore, $Q_{12}$ is of the desired form (\ref{eq:EVENrDERIVATIVE+2-1})
implying $Q_{12}(0)=0$. Consider now the second term in (\ref{eq:Tm2DERHATQ2}).
Put
\begin{equation}
Q_{22}:=[\frac{\partial^{l+1}f_{j}(y(t))}{\partial y_{k_{l+1}}\partial y_{k_{l}}\ldots\partial y_{k_{2}}\partial y_{k_{1}}}]y_{k_{l+1}}^{(1)}J_{O}J_{E}^{(1)}.\label{eq:Q22}
\end{equation}
First we scrutinize the expression $J_{E}^{(1)}$ . If $J_{E}\equiv1$
then $Q_{22}\equiv0$ and trivially $Q_{22}(0)=0$ as desired. If
$w\geq1$ then $J_{E}^{(1)}$ is the sum of $w$ products as follows

\[
J_{E}^{(1)}=y_{z_{1}}^{(2c_{1}+1)}(t)y_{z_{2}}^{(2c_{2})}(t)\ldots y_{z_{w}}^{(2c_{w})}(t)+y_{z_{1}}^{(2c_{1})}(t)y_{z_{2}}^{(2c_{2}+1)}(t)\ldots y_{z_{w}}^{(2c_{w})}(t)
\]
\begin{equation}
+\ldots+y_{z_{1}}^{(2c_{1})}(t)y_{z_{2}}^{(2c_{2})}(t)\ldots y_{z_{w}}^{(2c_{w}+1)}(t),\qquad c_{1},c_{2},\ldots,c_{w},w\in\mathbb{N}.\label{eq:JEDER1storder-1}
\end{equation}
Without loss of generality assume that a representative product in
(\ref{eq:JEDER1storder-1}) has the form 
\begin{equation}
F_{22}:=y_{z_{1}}^{(2c_{1}+1)}(t)y_{z_{2}}^{(2c_{2})}(t)\ldots y_{z_{w}}^{(2c_{w})}(t).\label{eq:F22}
\end{equation}
Combine $y_{k_{l+1}}^{(1)}$ from (\ref{eq:Q22}) with the factor
$y_{z_{1}}^{(2c_{1}+1)}$ in (\ref{eq:F22}) and put 
\begin{equation}
\widehat{J}_{O}:=y_{k_{l+1}}^{(1)}y_{z_{1}}^{(2c_{1}+1)}J_{O},\;\widehat{J}_{E}:=y_{z_{2}}^{(2c_{2})}(t)\ldots y_{z_{w}}^{(2c_{w})}(t).\label{eq:2NDTERMDERGREEKJODERJE-1}
\end{equation}
Evidently, $\widehat{J}_{O}$ has an even number $(r+2)$ of odd order
derivatives. However, the number of \,factors in $\widehat{J}_{E}$
is now $q=w-1$. Recall, that if $l$ is an even number then $s=l+1$
is an odd number. If $l$ is an odd number then\, $s=l+1$ is an
even number. Therefore, conditions i) and ii) in (\ref{eq:EVENrDERIVATIVE+2-1})
hold. Hence (\ref{eq:EVENrDERIVATIVE+2-1}) is of the desired form
and $\widehat{T_{m}^{2}}(0)=0$. We are left to analyze the terms
in (\ref{eq:Tm2HATQ3}). We first focus on a representative product
in the middle term of (\ref{eq:Tm2HATQ3}) that is in (\ref{eq:Q23})
below
\begin{equation}
Q_{23}:=\frac{\partial^{l}f_{j}(y(t))}{\partial y_{k_{l}}\ldots\partial y_{k_{2}}\partial y_{k_{1}}}J_{O}^{(1)}J_{E}^{(1)}.\label{eq:Q23}
\end{equation}
We recall the forms $J_{O}^{(1)},\;J_{E}^{(1)}$ in (\ref{eq:JODERIVATive-1})
and (\ref{eq:JEDER1storder-1}) respectively. There are $rw$ terms
in $J_{O}^{(1)}J_{E}^{(1)}$ . Each term is a products of $rw$ factors.
We focus on each product. Assume that $l$ is even. If $J_{E}\equiv1$
or $w=0$ then 
\begin{equation}
J_{O}^{(1)}J_{E}^{(1)}\equiv0\Longrightarrow Q_{23}\equiv0\Longrightarrow\frac{\partial^{l}f_{j}(y(0))}{\partial y_{k_{l}}\ldots\partial y_{k_{2}}\partial y_{k_{1}}}J_{O}^{(1)}(0)J_{E}^{(1)}(0)=0,\label{eq:Q23=00003D0}
\end{equation}
Notice that $s=l$. Therefore, we may assume without loss of generality
that a representaive of one of these $rw$ products has the form

\begin{equation}
J_{OS}J_{ES}:=y_{s_{1}}^{(2u_{1}+1)}y_{s_{2}}^{(2u_{2}+1)}\ldots y_{s_{r-1}}^{(2u_{r-1}+1)}y_{s_{r}}^{(2u_{r}+2)}y_{z_{1}}^{(2c_{1}+1)}(t)y_{z_{2}}^{(2c_{2})}(t)\ldots y_{z_{w}}^{(2c_{w})}(t).\label{eq:WLOGPRODUCTJODERJEDER-1}
\end{equation}
Put 
\begin{equation}
\widehat{J}_{O}:=y_{s_{1}}^{(2u_{1}+1)}y_{s_{2}}^{(2u_{2}+1)}\ldots y_{s_{r-1}}^{(2u_{r-1}+1)}y_{z_{1}}^{(2c_{1}+1)}(t),\;\widehat{J}_{E}:=y_{s_{r}}^{(2u_{r}+2)}(t)y_{z_{2}}^{(2c_{2})}(t)\ldots y_{z_{w}}^{(2c_{w})}(t).\label{eq:SEPARATEDPRODUCTWLOGJOJEhats-1}
\end{equation}
One can easily verify that $\widehat{J}_{O}$ and $\widehat{J}_{E}$
have the same number of factors as $J_{O}$ and $J_{E}$ respectively.
Since $s=l$, $\widehat{J}_{O}$ and $\widehat{J}_{E}$ in (\ref{eq:SEPARATEDPRODUCTWLOGJOJEhats-1})
are of the desired form (\ref{eq:EVENrDERIVATIVE+2-1}). We are left
to analyze the representative products in the first and third term
in (\ref{eq:Tm2HATQ3}). They are
\begin{equation}
Q_{13}:=\frac{\partial^{l}f_{j}(y(t))}{\partial y_{k_{l}}\ldots\partial y_{k_{2}}\partial y_{k_{1}}}J_{O}^{(2)}J_{E};\quad Q_{33}:=\frac{\partial^{l}f_{j}(y(t))}{\partial y_{k_{l}}\ldots\partial y_{k_{2}}\partial y_{k_{1}}}J_{O}J_{E}^{(2)};\label{eq:Q13=000026Q33}
\end{equation}
To this end we scrutinize the factors that make up the products in
$J_{O}J_{E}^{(2)}$ and in $J_{O}^{(2)}J_{E}$. We start with $J_{O}J_{E}^{(2)}$
. If $w=0$ we have $Q_{33}\equiv0$. If $w=1$ then $s=l$ where
$s$ and $l$ are both odd numbers and we put 
\begin{equation}
\widehat{J}_{O}=J_{O},\quad\widehat{J}_{E}=J_{E}^{(2)}=y_{z_{1}}^{(2c_{1}+2)}.\label{eq:JOJEHATSAGAINQ33}
\end{equation}
 Evidently, conditions i) and ii) are satisfied in (\ref{eq:EVENrDERIVATIVE+2-1})
. Assume that $w\geq2$. Then $J_{E}^{(2)}$ is a sum of products
that are of two kinds. The first kind is

\begin{equation}
y_{z_{j}}^{(2c_{j}+2)}\prod_{l\neq j,l=1}^{w}y_{z_{l}}^{(2c_{l})}.\label{eq:JOJODer2-1}
\end{equation}
Then we put
\begin{equation}
\widehat{J}_{O}=J_{O},\;\widehat{J}_{E}=y_{z_{j}}^{(2c_{j}+2)}\prod_{l\neq j,l=1}^{w}y_{z_{l}}^{(2c_{l})}.\label{eq:JODer2JEALPHA=00003D2-1-2}
\end{equation}
 We have again that $\widehat{J}_{O}$ and $\widehat{J}_{E}$ have
the same number of factors as $J_{O}$ and $J_{E}$ respectively.
Moreover, $s=l$. Thus, conditions i) and ii) are satisfied in (\ref{eq:EVENrDERIVATIVE+2-1}).
The second kind of products is

\begin{equation}
J_{E}^{(2)}=y_{z_{j}}^{(2c_{j}+1)}y_{z_{k}}^{(2c_{k}+1)}\prod_{l\neq j,k,l=1}^{w}y_{z_{l}}^{(2c_{l})},\quad\prod_{l\neq j,k,l=1}^{w}y_{z_{l}}^{(2c_{l})}:\equiv1\quad if\quad w=2.\label{eq:JOAlphaj=00003D1Alphak=00003D1-2-1-2}
\end{equation}
Put

\begin{equation}
\widehat{J}_{O}:=y_{z_{j}}^{(2c_{j}+1)}y_{z_{k}}^{(2c_{k}+1)}J_{O},\quad\widehat{J}_{E}:=\prod_{l\neq j,k,l=1}^{w}y_{z_{l}}^{(2c_{l})}.\label{eq:JEAlphaj,k=00003D1andw>2-2-2}
\end{equation}
Evidently, with s$=l$ the parity of $s,l,$ and the number of factors
in $\widehat{J}_{E}$ that is $(w-2)$ is the same. If $l,w$ are
both even numbers so are $s,w-2$. If $l,w$ are both odd numbers
so are $s,w-2$. The number of factors in $J_{O}$ is $r$ and the
number of factors in $\widehat{J}_{O}$ is $(r+2)$ as needed. Thus,
conditions i) and ii) are satisfied in (\ref{eq:EVENrDERIVATIVE+2-1}).
We are left to conclude the inductive process by focusing on the first
term $Q_{13}$ in (\ref{eq:Q13=000026Q33}). To this end we focus
on the two kinds of products that emanate in the second derivative
$J_{O}^{(2)}$ . $J_{O}$ has the form 
\begin{equation}
J_{O}=y_{s_{1}}^{(2e_{1}+1)}(t)y_{s_{2}}^{(2e_{2}+1)}(t)\ldots y_{s_{r}}^{(2e_{r}+1)}(t),\qquad e_{1},e_{2},\ldots,e_{r}\in\mathbb{N}_{0}.\label{eq:JOAgain-1-2-1}
\end{equation}
The first kind of second derivatives $J_{O}^{(2)}$ will contain the
following $r$ products
\[
[y_{s_{1}}^{(2e_{1}+3)}(t)y_{s_{2}}^{(2e_{2}+1)}(t)\ldots y_{s_{r}}^{(2e_{r}+1)}(t)],\;[y_{s_{1}}^{(2e_{1}+1)}(t)y_{s_{2}}^{(2e_{2}+3)}(t)\ldots y_{s_{r}}^{(2e_{r}+1)}(t)]
\]
\begin{equation}
,\ldots,[y_{s_{1}}^{(2e_{1}+1)}y_{s_{2}}^{(2e_{2}+1)}(t)\ldots y_{s_{j}}^{(2e_{j}+3)}\ldots y_{s_{r}}^{(2e_{r}+1)}(t)]\ldots[y_{s_{1}}^{(2e_{1}+1)}(t)y_{s_{2}}^{(2e_{2}+1)}(t)\ldots y_{s_{r}}^{(2e_{r}+3)}(t)].\label{eq:JODer2Leibnitzgeneral-3-1}
\end{equation}
Consequently, a representative product in $J_{O}^{(2)}J_{E}$ will
have the form

\[
[y_{s_{1}}^{(2e_{1}+1)}y_{s_{2}}^{(2e_{2}+1)}(t)\ldots y_{s_{j}}^{(2e_{j}+3)}\ldots y_{s_{r}}^{(2e_{r}+1)}(t)]y_{k_{1}}^{(2c_{1})}(t)y_{k_{2}}^{(2c_{2})}(t)\ldots y_{k_{w}}^{(2c_{w})}(t).
\]
Choose 
\begin{equation}
\widehat{J}_{O}=[y_{s_{1}}^{(2e_{1}+1)}y_{s_{2}}^{(2e_{2}+1)}(t)\ldots y_{s_{j}}^{(2e_{j}+3)}\ldots y_{s_{r}}^{(2e_{r}+1)}(t)],\;\widehat{J}_{E}=J_{E}.\label{eq:JODer2JEALPHA=00003D2-2-1}
\end{equation}
Evidently, $\widehat{J}_{O}$ and $\;\widehat{J}_{E}$ have the same
number of factors as $J_{O}$ and $J_{E}$ respectively. Since $s=l$,
$\widehat{J}_{O}$ and $\widehat{J}_{E}$ are of the desired form
(\ref{eq:EVENrDERIVATIVE+2-1}). One may assume without loss of generality
that the second kind of products in $J_{O}^{(2)}$ is 
\begin{equation}
y_{s_{1}}^{(2e_{1}+2)}(t)y_{s_{2}}^{(2e_{2}+2)}(t)y_{s_{3}}^{(2e_{1}+1)}\ldots y_{s_{r}}^{(2e_{r}+1)}(t).\label{eq:JOSECONDder1}
\end{equation}
In general we will have a product of the form 
\begin{equation}
y_{s_{j}}^{(2e_{j}+2)}y_{s_{k}}^{(2e_{k}+2)}\quad if\quad r=2,\quad y_{s_{j}}^{(2e_{j}+2)}y_{s_{k}}^{(2e_{k}+2)}\prod_{l\neq j,k}^{r}y_{s_{l}}^{(2e_{l}+1)},\quad l=1,2,\ldots r,\;r\geq3.\label{eq:JOAlphaj=00003D1Alphak=00003D1-3-1}
\end{equation}
Then put:
\begin{equation}
\widehat{J}_{O}\equiv1,\quad\widehat{J}_{E}=y_{s_{j}}^{(2e_{j}+2)}y_{s_{k}}^{(2e_{k}+2)}J_{E},\quad if\quad r=2,\label{eq:JOAlphaj,k=00003D1,1r=00003D2-2-1}
\end{equation}

\begin{equation}
\widehat{J}_{O}:=\prod_{l\neq j,k,l=1}^{r}y_{s_{l}}^{(2e_{l}+1)},\quad\widehat{J}_{E}:=y_{s_{j}}^{(2e_{j}+2)}y_{s_{k}}^{(2e_{k}+2)}J_{E},\quad l=1,2,\ldots r,\;r\geq3.\label{eq:JOAlphaj,k=00003D1andr>2-2-1}
\end{equation}
It is readily observed that $\widehat{J}_{O},\widehat{J}_{E}$, are
of the desired form (\ref{eq:EVENrDERIVATIVE+2-1}). Hence, all $T_{m}(0)=0$
imply that all $\widehat{T_{m}^{(2)}}(0)=0$. Consequently, $y^{(m)}(0)=\overrightarrow{0}$
for all even numbers $m$$\in\mathbb{N}_{0}$ . Q.E.D.
\end{proof}

\end{document}